\documentclass{eptcs}
\usepackage[mathletters]{ucs}
\usepackage[utf8x]{inputenx}
 \usepackage{amsmath, xypic, graphicx}
\usepackage{stmaryrd,xspace,amssymb}
\usepackage[active]{srcltx}
\usepackage{url}

\newcommand{\covered}{\ensuremath{\mathop{\vartriangleleft}}}

\newcommand{\beq}{\begin{equation}}
\newcommand{\eeq}{\end{equation}}
\newcommand{\bea}{\begin{eqnarray}}
\newcommand{\eea}{\end{eqnarray}}

\newcommand{\Sets}{\mbox{\textbf{Sets}}}

 \newcommand{\cov}{\nabla}

\newcommand{\id}[1]{\ensuremath{\mathrm{id}}}
\newcommand{\op}{\ensuremath{^{\mathrm{op}}}}

\newcommand{\Sh}{\ensuremath{\mathrm{Sh}}}

\newcommand{\uS}{\underline{\Sigma}}
\newcommand{\uA}{\underline{A}}

\newcommand{\ie}{\textit{i.e.}}

\newtheorem{theorem}{Theorem}
\newtheorem{lemma}[theorem]{Lemma}

\newtheorem{corollary}[theorem]{Corollary}
\newtheorem{definition}[theorem]{Definition}

\newenvironment{proof}[1][Proof]%
{ \begin{trivlist}%
  \item[\hskip \labelsep {\bfseries #1}]%
}%
{ \end{trivlist}%
}
\newcommand{\qed}{\nobreak\hfill$\Box$}

\usepackage[active]{srcltx}
\renewcommand{\cov}{\covered}

\newcommand{\cC}[1]{\ensuremath{\mathcal{C}(#1)}}
\def\titlerunning{The space of measurement outcomes}
\def\authorrunning{Bas Spitters}
\begin{document}

\title{The space of measurement outcomes as a spectrum for non-commutative algebras}
\author{Bas Spitters}
\maketitle
\begin{abstract}
Bohrification defines a locale of hidden variables internal in a topos. We
find that externally this is the space of \emph{partial}
measurement outcomes. By considering the $\neg\neg$-sheafification, we obtain the space of
measurement outcomes which coincides with the spectrum for commutative C*-algebras.
\end{abstract}

\section{Introduction}
By combining Bohr's philosophy of quantum mechanics, Connes' non-commuta\-tive
geometry~\cite{Connes}, constructive Gelfand
duality~\cite{banaschewskimulvey00a,banaschewskimulvey00b,coquand05,CoquandSpitters:cstar} and
inspiration from Doering and Isham's
spectral presheaf~\cite{doeringisham:review}, we proposed Bohrification
as a spatial quantum logic~\cite{qtopos,Bohrification}. Given a C*-algebra $A$, modeling a quantum
system, consider the
poset of Bohr's classical concepts\[
\cC A:=\{C  \mid C\text{ is a commutative C*-subalgebra of }A\}.
\]
In the functor topos $\Sets^{\cC A}$ we consider the \emph{Bohrification} $\uA$: the trivial functor
$C\mapsto C$.
This is an internal C*-algebra of which we can compute the spectrum, an internal locale Σ in the
topos $\Sets^{\cC A}$. This locale, or its externalization, is our proposal for an intuitionistic
quantum logic~\cite{qtopos,Bohrification}. 

In the present paper we explore a possible refinement of this proposal motivated by what happens 
for commutative algebras and by questions about maximal subalgebras. 

In section~\ref{pMO} we compute the externalization of this locale. It is the
space of partial measurement outcomes: the points are pairs of a C*-subalgebra together with a
point of its spectrum. This construction raises two natural questions:
\begin{itemize}
\item Can we restrict to the maximal commutative subalgebras, i.e.\ total measurement frames?
\item Are we allowed to use classical logic internally?
\end{itemize}
In section~\ref{maximal} we will see that, in a sense, the answers to both of these questions
are positive. The collection of maximal commutative subalgebras covers the space in
the dense topology and this dense, or double negation, topology forces (sic) the
logic to be classical. By considering the $\neg\neg$-sheafification, we obtain a genuine
generalization of the spectrum. Moreover, our previous constructions~\cite{qtopos} of the phase
space (Σ) and the state space still apply essentially unchanged.

\section{Preliminaries}
An extensive introduction to the context of the present paper can be found
in~\cite{Bohrification,qtopos} and the references therein. Here we will just repeat the bare minimum
of definitions.

A site on an poset defines a covering relation. To simplify the presentation we restrict to the
case of a meet-semilattice.

\begin{definition}\label{def:cover}
  Let $L$ be a meet-semilattice. A \emph{covering relation} on $L$
  is a relation $\cov ⊂  L × P(L)$ satisfying:
  \begin{enumerate}
    \item if $x ∈ U$ then $x \cov U$;
    \item if $x \cov U$ and $U \cov V$ (\ie\ $y \cov V$ for all
      $y \in U$) then $x \cov V$;
    \item if $x \cov U$ then $x \wedge y \cov U$;
    \item if $x \cov U$ and $x \cov V$, then $x \cov U \wedge V$,
      where $U \wedge V = \{x \wedge y \mid x \in U, y \in V\}$.
  \end{enumerate}
Such a pair $(L,\cov)$ is called a \emph{formal topology}.
\end{definition}

 Every formal topology defines a locale, conversely every locale can be presented in such a way.

\begin{definition}
  Let $(L, \cov)$ be a formal topology. A \emph{point} is an inhabited 
  $α ⊂ L$ that is filtering with respect to $≤ $, and such
  that for each $a ∈ α $ if $a \cov U$, then $U ∩ α$ is
  inhabited. In short, it is a completely prime filter.
\end{definition}

The spectrum Σ of a C*-algebra $A$ can be described directly as a lattice $L(A)$ together with covering a relation; see~\cite{CoquandSpitters:cstar}.

\section{Measurements}
In algebraic quantum theory~\cite{Emch,Haag:LQP,landsman98}, a measurement is a (maximal) Boolean subalgebra of the set of projections of a von Neumann algebra. The outcome of a measurement is the consistent assignment of either 0 or 1 to each element (test, proposition) of the Boolean algebra: the outcome is an element of the Stone spectrum. Unlike von Neumann-algebras, C*-algebras need not have enough projections. It is customary to replace the Boolean algebra by a commutative C*-subalgebra and the Stone spectrum by the Gelfand spectrum. Although a detailed critique of the measurement problem is beyond the scope of this paper, with the previous motivation we will make the following definition.

\begin{definition}
A \emph{measurement outcome} is a point in the spectrum of a maximal commutative subalgebra.
\end{definition}

The choice to restrict to \emph{maximal} subalgebras varies between authors. The present choice fits with our presentation.

\section{The space of partial measurement outcomes}\label{pMO}
Iterated topos constructions, similar to iterated forcing in set theory were studied by
Moerdijk~\cite{Moerdijk}\cite[C.2.5]{johnstone02a}. To wit, let $\mathcal{S}$ be the
ambient topos. One may think of the topos \Sets, but we envision applications where a different
choice for $\mathcal{S}$ is appropriate~\cite{Bohrification}.

\begin{theorem}[Moerdijk] Let $\mathbb{C}$ be a site in $\mathcal{S}$ and
$\mathbb{D}$ be a site in $\mathcal{S}[\mathbb{C}]$, the topos of sheaves over $\mathbb{C}$. Then
there is a site\footnote{The notation $\ltimes$ is motivated by the special case where $\mathbb{C}$
is a group $G$ considered as a category with one object and $\mathbb{D}$ is a
group $H$ in $\Sets^G$. Then $\mathbb{C}\ltimes \mathbb{D}$ is indeed the
semi-direct product $H\ltimes G $}
$\mathbb{C}\ltimes\mathbb{D}$ such that
\[
\mathcal{S}[\mathbb{C}][\mathbb{D}]=\mathcal{S}[\mathbb{C}\ltimes\mathbb{D}].
\]
\end{theorem}

We will specialize to sites on a poset and without further ado focus on our main example.
As before, let \[\cC A:=\{C  \mid C\text{ is a commutative C*-subalgebra of }A\}.\]
Let $\mathbb{C}:=\cC A\op$ and $\mathbb{D}=Σ$ the spectrum of the Bohrification, we
compute $\mathbb{C}\ltimes \mathbb{D}$. The objects are pairs $(C,u)$, where $C\in \cC A$ and $u$
in $L(C)$. Define the order $(D,v)\leq (C,u)$ as $D ⊃ C$ and $v ⊂ u$.
In terms of
forcing, this is the information order and the objects are forcing conditions.
We add a covering relation $(C,u)\cov (D_i,v_i)$ as for all $i$, $C⊂ D_i$ and $C \Vdash u\cov V$,
where $V$ is the pre-sheaf generated by the conditions $D_i\Vdash v_i\in V$. It follows from the general theory that this
is a Grothendieck topology.

We simplify: the pre-sheaf $V$ is generated by the
conditions $D_i\Vdash v_i ∈ V$ means 
\[V(D):=\{v_i ∈ V(D) \mid D ⊃ D_i\}.\] 
Hence,
\[C\Vdash u\cov V \text{ iff }u\cov\{v_i\mid D_i=C\}\]
by the following lemma.

\begin{lemma}\emph{\cite{qtopos}}
Let $V$ be an internal sublattice of $L$. Then $C\Vdash u\cov V$ iff $u\cov V(C)$.
\end{lemma}

\begin{definition}
A \emph{partial measurement outcome} is a point in the spectrum of a commutative subalgebra.
A \emph{consistent ideal of partial measurement outcomes} is a family $(C_i,σ_i)$ of partial measurement outcomes such that the $C_i$ are an ideal in $\cC A$ and if $C_i ⊂ C_j$, then $σ_i=σ|_{C_j}$.
\end{definition}

\begin{theorem}\label{thm:pMO}
The points of the locale generated by $\mathbb{C}\ltimes\mathbb{D}$ are consistent ideals of partial measurement outcomes.
\end{theorem}
\begin{proof}
Let τ be a point, that is a completely prime filter. Suppose that $(D,u) ∈ τ$, then by the
covering relation for $\mathbb{D}$, τ defines a point of the spectrum $Σ(D)$. This point is defined
consistently: If $u∈ L(C) ⊂ L(D)$, then $(D,u)\leq (C,u)$. Hence, if $(D,u) ∈ τ$, so is
$(C,u)$ and the point in $Σ(D)$ defines a point in $Σ(C)$ as a restriction of functionals.
When both $(C,u)$ and $(C',u')$ are in τ, then, by directedness, there exists $(D,v)$ in τ such that
$C,C'⊂ D$ and $v⊂ u,u'$. Moreover, $(C,u)\in τ$ implies $(C,\top)∈ τ$, now the set $\{C\mid (C,\top)\in τ\}$ is directed and downclosed.

Conversely, let $(C_i,σ_i)$ be a consistent ideal of partial measurement outcomes, then
\[\mathcal{F}:=\{(C_i,u)\mid  σ_i ∈ u\}\]
defines a filter: it is up-closed and lower-directed. Suppose that $(C_i,u)\cov (A_j,v_j)$, that is $j$, $C⊂ A_j$ and $A_k \Vdash u\cov V$, where $V=\{ v_j\mid A_j=A_k\}$. Then $(A_k,v_j)∈ \mathcal{F}$, because $σ_k$ is a point/completely prime filter.
\qed\end{proof}

It is tempting to identify the ideal of partial measurement outcomes with its limit. However, the ideal and its limit define different points. These points are identified in Section~\ref{maximal}.

Let us call this locale $pMO$ for (consistent ideals of) partial measurement outcomes. For commutative
C*-algebras $pMO$ is similar, but not equal, to the spectrum:
\begin{corollary}\label{spectrum}
For a compact regular $X$, the points of $pMO(C(X))$ are points of
the spectrum of a C*-\emph{sub}algebra of $C(X)$.
\end{corollary}

An explicit external description of the locale may be found in~\cite{Bohrification}. The present
computation gives an alternative description which makes it easy to compute the
points.

\section[Maximal subalgebras and classical logic]{Maximal commutative subalgebras, classical logic
and the spectrum}\label{maximal}
As stated in the introduction, we address the following questions:
\begin{itemize}
\item Can we restrict to the \emph{maximal} commutative subalgebras?
\item Are we allowed to use classical logic internally?
\end{itemize}

In a sense, the answers to both of these questions are positive. The collection of maximal
commutative subalgebras covers the space $\cC A$ in the dense topology and this dense, or double
negation, topology forces the logic to be classical.

Sheaves for the dense topology may be used to present classical set theoretic forcing or Boolean
valued models. In set theoretic forcing one considers the topos
$\Sh(P,\neg\neg)$~\cite[p.277]{maclanemoerdijk92}. The dense topology on a poset $P$ is defined as
$p\cov D$ if
$D$ is dense below $p$: for all $q\leq p$, there
exists a $d\in D$ such that $d\leq q$.\footnote{Constructively, this also defines a topology~\cite{Coquand:Goodman}. However, we need classical logic to prove that it coincides with the double negation topology.} The locale presented by this site is a Boolean algebra,
the topos is a Boolean valued model. If $P$ is directed and $I$ is an ideal in $P$, the directed join is contained in the double negation of $I$.
Hence the double negation topology is coarser than the $j$-topology described above.

This topos of $\neg\neg$-sheaves satisfies the axiom of
choice~\cite[VI.2.9]{maclanemoerdijk92} when our base topos does. The associated sheaf functor sends
the presheaf topos $\hat{P}$ to the sheaves $\Sh(P,\neg\neg)$. The sheafification can be described
explicitly~\cite[p.273]{maclanemoerdijk92} for $V\rightarrowtail W$:
\[
\neg\neg V(p)=\{x\in W(p)\mid \text{for all } q\leq p \text{ there
exists } r\leq q \text{ such that } x\in V(r) \}.
\]

We apply this to the poset $\cC A$.
We write $A$ for the constant functor $C\mapsto A$. Then $\uA\subset A$ in $\Sets^{\cC A}$.

For commutative $A$, $\cC A$ has $A$ as bottom element. For all $C$, $\uA_{\neg\neg}(C)=A$.

For the general case, we observe that each $C$ is covered by the collection of all its
supersets. By Zorn\footnote{Here we use classical meta-logic.}, each commutative subalgebra is
contained in a maximal commutative one. Hence the
collection of maximal commutative subalgebras is dense. So, $\uA_{\neg\neg}(C)$
is the intersection of all maximal commutative subalgebras containing $C$.

The covering relation for $(\cC A,\neg\neg)⋉\uS$ is $(C,u)\cov (D_i,v_i)$ iff $C ⊂ D_i$ and $C⊩
u\cov V_{\neg\neg}$, where $V_{\neg\neg}$ is the sheafification of the presheaf $V$ generated by the
conditions $D_i⊩v_i∈ V$. Now, $V\rightarrowtail L$, where $L$ is the spectral lattice of the
\emph{presheaf} $\uA$.
\begin{eqnarray*}
V_{\neg\neg}(C) &=&\{u∈ L(C)∣ ∀ D≤ C ∃ E≤ D . u∈ V(E)\}.
\end{eqnarray*}
So, $(C,u)\cov (D_i,v_i)$ iff
\[
∀ D≤ C ∃ D_i≤ D. u\cov V(D_i).
\]

\begin{theorem}
The locale $MO$ generated by $(\cC A,\neg\neg)⋉\uS$ classifies measurement outcomes.  It is a
(dense) sublocale of $pMO$.
\end{theorem}
\begin{proof}
 In the context of Theorem~\ref{thm:pMO} we suppose that $(C,\top)∈τ$. The subalgebra $C$ is
covered by all the maximal commutative subalgebras containing it, so by directedness we conclude
that $(M,\top)∈ τ$ for some maximal $M$.
\end{proof}

The $MO$ construction is a non-commutative generalization of the
spectrum. In this sense it behaves better then $pMO$; compare
Corollary~\ref{spectrum}.
\begin{corollary}
For a compact regular $X$, $X ≅ MO(C(X))$.
\end{corollary}
\begin{proof}
$C(X)$ is the only maximal commutative subalgebra of $C(X)$.
\qed\end{proof}

By considering the double negation we may use classical logic internally in our Boolean valued
model. 

The Kochen-Specker theorem can be reformulated as the non-existence of certain global sections~\cite{doeringisham:review,Bohrification}.
This connection carries over essentially unchanged, since a global section of a sheaf is (by definition) a global section of the sheaf considered as a presheaf.
\begin{theorem}Kochen-Specker: Let $H$ be a Hilbert space with $\dim H>2$ and let $A=B(H)$.
Then the $\neg\neg$-sheaf ∑ does not allow a global section.
\end{theorem}
Internally the axiom of choice holds, so Σ is a compact Hausdorff space. Still
spectrum does not have a \emph{global} 
point and the algebra does not have a global element. 

As an example, consider the matrix algebra $M_n$. Let $D_n$ be the $n$-dimensional diagonal matrix.
The maximal
subalgebras of $M_n$ are $\{φ D_n ∣ φ ∈ SU_n\}$; see \cite{n-level}. Moreover, $Σ ≅ \{1,\ldots,
n\}$
in
$\Sh(\cC A,\neg\neg)$. This is a complete Boolean algebra.
We have arrived at the setting of iterated forcing as in set theory.
Iterated forcing in set theory may be presented as follows; see Moerdijk~\cite[Ex~1.3a]{Moerdijk}.
If $P$ is a poset in $\Sets$, and $Q$ is a poset
in $\hat P$, then $P ⋉ Q$ is the poset in \Sets\ of pairs $(p,q)$ with $p ∈ P$, $p ⊩ q
∈ Q$, and $( p, q) ≤ (p', q')$ iff $p≤ p'$ and $p⊩ q≤ q'$. If
$\mathcal{E}=\Sh(P,\neg\neg)$, and $Q$ is a poset in $\mathcal{E}$, $\mathcal{F}:=
\Sh_\mathcal{E}(Q,\neg\neg)$, then $\mathcal{F}\cong \Sh(P⋉ Q,\neg\neg)$. In other words,
$(P,\neg\neg)⋉ (Q,\neg\neg)≅(P\ltimes Q,\neg\neg)$.
If, as in the case of $M_n$, $Q$ is a cBA in $\mathcal{E}$, then $(Q,\neg\neg)\cong Q$. So $(P⋉
Q,\neg\neg)≅ (P,\neg\neg)⋉ Q$. We expect similar simplifications when starting from a Rickart
C*-algebra~\cite{HLSSyn}.

A similar $\neg\neg$-transformation can be applied to our Bohrification of OMLs. In the
example studied in~\cite{Bohrification}, we compute a 17 element \emph{Heyting} algebra from an OML.
Adding
the double negation we obtain a 16 element \emph{Boolean} algebra. The function
$f(0)=0$ and $f(i)=1$ is `eventually' equal to the constant function 1.
As a result, we obtain the product of 4 Boolean algebras, the spectrum is the coproduct of the
corresponding locales.

\section{Conclusions and further research}
We have presented a non-commutative generalization of the spectrum motivated by physical
considerations.

We suggest another way to restrict to maximal subalgebras, while preserving the possibility to
compute a unique functional from a global section. Consider a matrix algebra. Let $p ∈ C $ be a
projection and suppose that $M \mapsto σ_M∈ Σ(M)$ is continuous with respect to the unitary group
action. Then $σ(p)∈ \{0,1\}$, say it is 0. Since the unitary group is connected and acts
transitively on the maximal subalgebras, $σ(u^*p)=0$ for all $u$. Suppose that $p∈ M_1,M_2$. Let $u$
transform $M_1$ into $M_2$, but leave $p$ fixed. We see that $σ(p)=0$ independent of the choice of
maximal subalgebra. By linearity and density, this extends from projections to general elements: σ
may be uniquely defined on all elements.
This suggests that, at least for matrix algebras, the independence guaranteed by the poset, may also
be guaranteed by the group action. We leave this issue to future research.

Bohrification, i.e.\ the $pMO$ construction, is not functorial when we equip
C*-algebras with their usual morphisms. The construction \emph{is} functorial
when we change the notion of morphism~\cite{vdBH}. More work seems to be need for the $MO$
construction: We have $(I_2, \top)\cov (C(2), \{(0,1), (1,0)\})$. However, this no longer holds when
we map $C(2)$
into $M_2$. In short, covers need not be preserved under natural notions of morphism.

Bohrification may be described as a (co)limit~\cite{vdBH}. While technically different the
intuitive meaning is similar: we are only interested in what happens eventually.

As in~\cite{qtopos,Bohrification} treat $\cC A$ as a mere poset. However, at
least in the finite dimensional case, this poset has an interesting manifold
structure~\cite{Bohrification}. 
Escardo~\cite{escardo1998properly} provides a construction of the support of a
locale which often coincides with its maximal points. It may be possible to use
this construction to refine the present results by maintaining the topological
structure.

\section{Acknowledgements} I would like to thank Chris Heunen, Klaas Landsman
and Steve Vickers for discussions
and comments on a draft of this paper. I thank the referees for their comments.

\bibliography{class,bibliography}
\bibliographystyle{alpha}
\end{document}